\documentclass[letterpaper]{article}
\usepackage{arxiv}
\usepackage[hidelinks,colorlinks,allcolors=magenta]{hyperref}
\usepackage{gra phicx}
\usepackage[numbers]{natbib}
\usepackage{caption}
\usepackage[dvipsnames,table]{xcolor}
\usepackage{amsmath,amsthm,amsfonts,amssymb}
\usepackage{booktabs} 

\usepackage[ruled]{algorithm2e} 

\SetAlFnt{\small}
\SetAlCapFnt{\small}
\SetAlCapNameFnt{\small}
\SetAlCapHSkip{0pt}
\IncMargin{-\parindent}

\usepackage{mathtools}
\usepackage[mathscr]{eucal}
\usepackage[nameinlink]{cleveref}
\usepackage{nth}
\usepackage{nicefrac}
\usepackage{tikz}
\usepackage{multirow}
\usepackage{rotating}
\usepackage{float}
\usepackage{pifont}

\usepackage{thm-restate}

\usepackage{pdfpages}
\renewcommand{\ge}{\geqslant}
\renewcommand{\geq}{\geqslant}
\renewcommand{\le}{\leqslant}
\renewcommand{\leq}{\leqslant}

\newcommand{\cmark}{\ding{51}}%
\newcommand{\xmark}{\ding{55}}%

\newcommand{\quest}{\textcolor{blue}{\textbf{?}}}

\newcommand{\ags}{\mathcal{N}}
\newcommand{\ag}{\ensuremath{i}}

\newcommand{\itms}{\mathcal{M}}
\newcommand{\itm}{\ensuremath{g}}
\newcommand{\itmm}{\ensuremath{g'}}

\newcommand{\mv}{v}
\newcommand{\av}{u}
\newcommand{\avs}{\vec{\av}}
\newcommand{\mvs}{\vec{\mv}}

\newcommand{\alc}{A}

\newcommand{\inst}{I}

\newcommand{\dd}{\mathrm{d}}

\newcommand{\sdp}{\succcurlyeq^{\!\mathsf{SD}}}

\definecolor{mohlatblue}{rgb}{0.36, 0.54, 0.66}
\definecolor{segreen}{RGB}{1, 121, 111}
\newcounter{note}[section]

\newcommand{\xx}{\mathbf{x}}
\newcommand{\pp}{\mathbf{p}}
\newcommand{\bbR}{\mathbb{R}_{\ge 0}}
\newcommand{\R}{\bbR}

\newcommand{\bb}{\ensuremath{\mathsf{bb}}}
\newcommand{\mbb}{\ensuremath{\mathsf{mbb}}}
\newcommand{\hierarchy}{\ensuremath{\mathsf{H}}}
\newcommand{\poly}{\ensuremath{\mathsf{poly}}}

\DeclareMathOperator*{\argmax}{arg\,max}
\DeclareMathOperator*{\argmin}{arg\,min}

\newtheorem{theorem}{Theorem}

\newtheorem{lemma}{Lemma}
\newtheorem{proposition}{Proposition}
\newtheorem{corollary}{Corollary}
\theoremstyle{definition}
\newtheorem{definition}{Definition}

\newtheorem{remark}{Remark}
\newtheorem{open}{Open Question}

\usepackage{tcolorbox}
\renewenvironment{open}{\refstepcounter{open}\begin{tcolorbox}[colback=gray!7, colframe=black, rounded corners] \vspace*{-4pt} \textbf{Open Question \theopen:}}{ \vspace*{-4pt}\end{tcolorbox}}

\DeclarePairedDelimiter{\set}{\{}{\}}

\DeclarePairedDelimiter{\ceil}{\lceil}{\rceil}

\newcommand{\wrt}{w.r.t.\xspace}

\newcommand{\pref}{\sigma}
\newcommand{\mpref}{\pi}

\newcommand{\calI}{\ensuremath{\mathcal{I}}}

\title{Fair Division with Market Values}
\author {
    Siddharth Barman\textsuperscript{\rm 1},
    Soroush Ebadian\textsuperscript{\rm 2},
    Mohamad Latifian\textsuperscript{\rm 2},
    Nisarg Shah\textsuperscript{\rm 2}
}
\date{
    \textsuperscript{\rm 1}Indian Institute of Science Bangalore\\
    \textsuperscript{\rm 2}University of Toronto\\
    barman@iisc.ac.in, \{soroush,latifian,nisarg\}@cs.toronto.edu
}

\begin{document}

\maketitle

\begin{abstract}
We introduce a model of fair division with market values, where indivisible goods must be partitioned among agents with (additive) subjective valuations, and each good additionally has a market value. The market valuation can be viewed as a separate additive valuation that holds identically across all the agents. We seek allocations that are simultaneously fair with respect to the subjective valuations and with respect to the market valuation. 

We show that an allocation that satisfies stochastically-dominant envy-freeness up to one good (SD-EF1) with respect to both the subjective valuations and the market valuation does not always exist, but the weaker guarantee of EF1 with respect to the subjective valuations along with SD-EF1 with respect to the market valuation can be guaranteed. We also study a number of other guarantees such as Pareto optimality, EFX, and  MMS. In addition, we explore non-additive valuations and extend our model to cake-cutting. Along the way, we identify several tantalizing open questions. 
\end{abstract}

\section{Introduction}
For centuries, human civilizations have pondered how to settle disputes centered around division of goods. Such resolution requirements come up often in cases such as divorce settlement and estate division (in the absence of a will). A common practice in the real world is not to liquidate all the goods but to assess their \emph{fair market value} (FMV) and create an ``equal division'' whereby the total fair market value of the goods awarded to each party is (almost) equal. Such a division is sometimes even a legal requirement \cite{Kagan21}. For instance, fair market values of land---as determined by assessors \cite{Kagan21b}--are often used to settle land disputes \cite{shtechman2021fair}. However, an FMV-based approach has two limitations: (i) parties may disagree on what the fair market value of a good is, and (ii) parties may have personal (i.e., subjective) cardinal preferences (over the goods) which may be quite different from the assessed market values.

Starting with the seminal work of \citet{Stein48}, the mathematical theory of fair division has made significant advances over the last century. The theory systematically addresses the issue (ii) mentioned above by developing models wherein the participating agents (parties) can submit their own \emph{heterogeneous} valuations over the goods. Here, the underlying goal is to find a division that each agent would consider fair \emph{according to their own valuation}. Hence, the models provide a systematic way to avoid disagreements. In recent years, fair division literature has made impressive strides in the context of indivisible goods, resulting in allocation methods that achieve appealing fairness guarantees. One such notion is \emph{envy-freeness up to one good} (EF1), which demands that no agent prefers the allocation of another agent over their own if we hypothetically exclude just one good from the envied agent's allocation~\cite{LMMS04,Bud11,CKMP+19}.

Despite the mathematical and conceptual appeal of this framework, equal division of goods under market values still holds unshakable importance in real-world dispute resolutions. Hence, we study the framework of \emph{fair division with market values}, where we seek an allocation of the goods that the agents find fair with respect to their heterogeneous personal preferences (``subjective utilities''), and that is also a near-equal division under the market values. Formally, a set $\itms$ of (indivisible) goods must be allocated among a set $\ags$ of $n$ agents. Each agent $i \in \ags$ submits an additive subjective utility function $u_i : 2^\itms \to \R$ and there is an additive market value function $v : 2^\itms \to \R$.\footnote{A set function $f : 2^\itms \to \R$ is said to be additive if $f(S) = \sum_{r \in S} f(\set{r})$ for all $S \subseteq \itms$ with $f(\emptyset) = 0$.} Our goal is to seek an allocation that is simultaneously fair (e.g., EF1) with respect to the subjective utilities $\set{u_i}_{i \in \ags}$ and with respect to the market values $v$. 

The significance of this formulation is further substantiated by the following observations:  First, near-equal division under the market values can correspond to a feasibility constraint, imposed due to legal or normative reasons. Hence, one can view the requirement under market values as seeking an allocation that is  fair (with respect to the agents' subjective utilities) as well as \emph{feasible} (with respect to the market values). Indeed, with this viewpoint, the current work contributes to the growing body of work on fair division under constraints (see, e.g., \cite{suksompong2021constraints}), and some of our results are derived using ideas from this line of work. 

Second, one can adopt an epistemic perspective that the subjective utilities provided by the agents are in fact their (personal) noisy estimates of the market values of the goods. Further, the ``market values'' that we take as input may be an expert's estimate or algorithmic prediction. Achieving an equal division with respect to both agents' personal estimates and the external estimate can provide greater robustness and prevent future disputes. 

Our work is closely related to two very recent papers, which also study fair division in the presence of market values.\footnote{Our work was conducted concurrently and independently of these works. That said, there is very little overlap between our results and the ones obtained in \cite{dall2023fair} and \cite{bu2023fair}, which we discuss at length in \Cref{sec:related}.}  \citet{dall2023fair} studies the same model, but with divisible goods and the subjective utility of each agent, for every good, modeled a function of the good's market value. The focus in \cite{dall2023fair} is on the behavior of prominent rules, rather than on concrete fairness guarantees. The work of \citet{bu2023fair} is closer to ours, who consider a generalization of our model with \emph{heterogeneous} market values $\set{v_i}_{i \in \ags}$, motivated as an allocator's preference over which agent should get which good. Their result for this general model establishes guarantees for proportionality, a notion weaker than the one studied in the current paper, namely, envy-freeness. In particular, there are two results where the current paper and the work of \citet{bu2023fair} overlap: First, we show that the fairness guarantee for goods obtained in \cite{bu2023fair} can be viewed as a corollary of a result due to \citet{BB18}. Second, we show that the existential result in \cite{bu2023fair} can be turned into a polynomial-time algorithm in our model (while \cite{bu2023fair} provides a polynomial-time algorithm for a weaker notion). Our study of Pareto optimality, EFX, and MMS in this framework is entirely novel. We discuss these technical points at length in \Cref{sec:related}. 

\subsection{Our Contributions}\label{sec:contributions}
In \Cref{sec:ef-both-sides}, we begin by asking whether there always exists an allocation of indivisible goods that satisfies \emph{stochastic-dominance envy-freeness up to one good} (SD-EF1)---a strengthening of EF1---\wrt\footnote{We will, throughout, abbreviate `with respect to' as \wrt} both additive subjective utilities and additive market values. We answer the question in the negative, using a carefully crafted instance with 2 agents and 7 goods. Complementing this result, we prove positive results for 2 agents and up to 6 goods, and for any number of agents with identical utilities. Then, relaxing the requirement slightly, we prove that an allocation that is EF1 \wrt subjective utilities and SD-EF1 \wrt market values always exists via a reduction to an algorithm by \citet{BB18}. The  complementary guarantee of SD-EF1 \wrt agents' utilities and EF1 \wrt market valuation stands as an interesting open question. 

Next, in \Cref{sec:extensions}, we consider other desiderata such as Pareto optimality (PO), maximin share fairness (MMS), and envy-freeness up to any good (EFX). In \Cref{sec:indiv-PO}, our main result is that an allocation that is PO \wrt subjective utilities and EF1 \wrt market values always exists, but there are open questions surrounding several other combinations of desiderata. \Cref{sec:mms-efx} shows that MMS or EFX cannot be attained \wrt one side (subjective utilities or market values) together with EF1, MMS, or EFX on the other side. Finally, in \Cref{sec:monotone}, we observe that the existence of an allocation that is EF1 \wrt even monotone subjective utilities and monotone market values remains open, but prove the existence of allocations that are $\nicefrac{1}{2}$-EF1 \wrt subadditive subjective utilities and SD-EF1 \wrt additive market values. These results are summarized in \Cref{tab:my-table}.

\newcommand{\cc}{\cellcolor{gray!25}}
\newcommand{\ccb}{\cellcolor{gray!10}}
\begin{table}[htbp]
    \centering
    \renewcommand{\arraystretch}{1.2}
    \begin{tabular}{cc|c@{\hskip 0.1in}c}
    \hline
    \multicolumn{2}{c|}{\multirow{2}{*}{\textbf{}}} & \multicolumn{2}{c}{\textbf{Market}} \\
    \multicolumn{2}{c|}{} & \textbf{SD-EF1} & \textbf{EF1} \\
    \hline
    \multirow{4}{*}{\begin{sideways}\textbf{Agents}\end{sideways}} & SD-EF1 & \ccb \textcolor{blue}{\xmark}~(\Cref{thm:sd-ef1-both-impossible}) & \ccb \quest \\
    & EF1 & \multicolumn{2}{c}{\cc\textcolor{blue}{\cmark}~(\Cref{cor:ef1-sdef1})} \\ 
    & PO & \ccb \textcolor{blue}{\xmark}~(\Cref{thm:po-sdef1-impossible}) & \ \ \ccb \textcolor{blue}{\cmark}~(\Cref{cor:ef1-fpo}) \\ 
    & MMS/EFX & \multicolumn{2}{c}{\cc \textcolor{blue}{\xmark}~(\Cref{thm:mms,thm:efx})}\\ 
    \hline
    \end{tabular}
    \caption{Summary of our results. PO on the market side is always guaranteed. Positive result with PO on the agents' side is conditioned on strictly positive utilities. MMS and EFX impossibilities hold even when swapping sides (i.e., EF1/SD-EF1 on the agents' side and MMS/EFX on the market side), or demanding MMS/EFX on both sides. 
    }
    \label{tab:my-table}
\end{table}

Finally, in \Cref{sec:cake}, we extend our setup of fair division with market values to cake-cutting. That is, we consider partitioning a heterogeneous divisible good (cake) among agents with subjective utilities in the presence of a market valuation, both subjective utilities and market valuation represented as countably additive measures over the cake. This model is formally introduced in \Cref{sec:cake}. We observe that the existence of an allocation that is EF \wrt both subjective utilities and market values can be reduced to the known existence of a \emph{perfect} cake-division~\cite{DS61}. We prove bounds on the number of cuts (of the cake) and queries (of the subjective utility measures and the market value measure) required to achieve EF \wrt both sides. We also show that there may \emph{not} exist any cake-division that is EF+PO \wrt the subjective utilities along with EF \wrt the market values. Though, relaxing EF+PO to simply PO \wrt the subjective utilities yields a positive result. 

Throughout, we identify a number of interesting and fundamental open questions.

\subsection{Related Work}\label{sec:related}

As mentioned previously, the works of \citet{dall2023fair} and \citet{bu2023fair} are the closest to ours. None of the results obtained in \citet{dall2023fair} imply anything in our model, since this prior work addresses {\it homogeneous} divisible goods. \citet{bu2023fair} study a generalization of our model with \emph{heterogeneous} market valuations $\set{v_i}_{i \in \ags}$.
They view $\set{v_i}_{i \in \ags}$ as the allocator's preferences over which agent should get which item. \citet{bu2023fair} provide four existential results:
\begin{enumerate}
    \item EF1 on both sides, when there is a single market valuation (i.e., our model). 
    \item EF1 on both sides for the particular case of two agents.
    \item PROP-$O(\log n)$ on both sides, where PROP (proportionality) is a notion weaker than envy-freeness.
    \item PROP-2 on both sides under binary valuations.
\end{enumerate}
The first result in the above-mentioned list is studied in the current work as well. However, we show that, in the case of goods (which is what \citeauthor{bu2023fair} address), the result is already implied by \cite{BB18}.
In the list above, the second result is strong, but it does not come with an efficient construction. In fact, they show how to achieve EF2 on both sides in polynomial time, which we are able to improve to EF1, on both sides, in polynomial time. Their EF1 existential result extends to monotone subjective utilities and monotone market values, whereas our EF1 result extends only to monotone subjective utilities and additive market values (though still with polynomial value query complexity). The last two results in the list above address the weaker notion of proportionality, which we do not focus on. In summary, the technical overlap between \cite{bu2023fair} and the current work is little. The two works together, however, point to interesting open questions, which we highlight throughout. 

A bit more broadly, fairness with respect to market values can be viewed as placing feasibility constraints (induced by the market values) on the allocation. Hence, fair division with market values can be considered as case of constrained fair division. Notably, many recent results address fair allocations subject to various feasibility constraints, such as cardinality constraints~\cite{BB18}, matroid constraints~\cite{dror2023fair}, or budget constraints~\cite{barman2023finding}. As we show in this work, EF1 with respect to an additive market valuation can be achieved by placing appropriately defined cardinality constraints. However, it is not clear if EF1 with respect to a more general (e.g., submodular) market valuation or with respect to heterogeneous additive market valuations can be reduced to known constraints.  

Our model is also related to recent works on fair division in two-sided markets~\cite{patro2020fairrec,FMS21,igarashi2023fair}, where agents on two sides of a market are matched to each other, with agents on each side having preferences over those on the other side. \citet{FMS21} also seek EF1 on both sides, but in their case the second side has its own allocation, whereas in our case, the market side is also evaluated using the same allocation. 


\section{Preliminaries}\label{sec:prelim}
An instance of fair division with market values is a quadruple, $\langle \ags, \itms, \mv, \avs=(\av_\ag)_{\ag \in \ags} \rangle$, where $\ags$ is a set of $n$ agents, $\itms$ is a set of $m$ indivisible goods, $\mv: 2^\itms \to \bbR$ is a monotone set function indicating the market valuation of the goods, and $\av_i: 2^\itms \to \bbR$ is the (subjective) monotone utility function of agent $\ag \in \ags$. In particular, $\av_i(S)$ and $\mv(S)$ denote agent $\ag$'s value and the market value for subset of goods $S \subseteq \itms$, respectively. Throughout most of the paper, we consider market valuation $\mv$ and utilities $\avs$ that are \emph{additive}: recall that a set function $f: 2^\itms \to \bbR$ is additive if $f(S) = \sum_{\itm \in S} f(\{\itm\})$ for all $S \subseteq \itms$. \Cref{sec:monotone}, in particular,  goes beyond additive utilities and considers general, monotone ones. For notational convenience, we will write $u_i(g) := u_i(\set{g})$ and $v(g) := v(\{g\})$. 

An allocation is $n$-partition $\alc = (\alc_1, \ldots, \alc_n)$ of $\itms$, where $\alc_i \subseteq \itms$ denotes the bundle of goods given to agent $\ag$. Here, $\alc_i \cap \alc_j = \emptyset$ for all $i \neq j$ and $\cup_{i\in \ags} \alc_i = \itms$. Our goal is to find allocations that are fair with respect to both subjective utilities and market valuation. A prominent fairness criterion for indivisible goods is envy-freeness up to one good. 

\begin{definition}[Envy-free up to one good]
For an instance $\langle \ags, \itms, \mv, \left( u_i \right)_{i \in \ags} \rangle$, we say that an allocation $\alc=(A_1, \ldots, A_n)$ is envy-free up to one good (EF1) with respect to subjective utilities if, for every pair of agents $i,j \in \ags$, either $\alc_j = \emptyset$ or  $\av_i(\alc_i) \ge \av_i(\alc_j \setminus \set{g})$ for some good $g \in \alc_j$. In addition, we say that $\alc$ is EF1 with respect to market values if the above-mentioned inequalities hold with $\av_i$ replaced by $\mv$. 
\end{definition}

Next, we define a strengthening, SD-EF1, based on the stochastic dominance (SD) relation.

\begin{definition}[SD-preference \cite{bogomolnaia2001new}]
For $X,Y \subseteq \itms$, we say that agent $\ag$ SD-prefers $X$ to $Y$, denoted $X \sdp_i Y$, if, for each good $\itm \in \itms$, it holds that $|\set{\itmm \in X : \av_i(\itmm) \geq \av_i(\itm)}| \geq  |\set{\itmm \in Y: \av_i(\itmm) \geq \av_i(\itm)}|$. That is, for each $g$, the subset $X$ has at least as many goods with utility (to agent $\ag$) at least as much as $\itm$ as $Y$ does. Also, we define SD-preference under the market valuation, denoted $X \sdp_\mv Y$, by replacing $\av_i$ with $\mv$. 
\end{definition}

\begin{definition}[SD-EF1~\cite{aziz2023best}]
An allocation $\alc=(A_1, \ldots, A_n)$ is said to be SD-envy-free up to one good (SD-EF1) if, for every pair of agents $i,j \in \ags$, either $\alc_j = \emptyset$ or $\alc_i \sdp_i \alc_j \setminus \set{g}$ for some good $g \in \alc_j$.
\end{definition}

We will say that an additive function $u: \itms \to \bbR$ induces a ranking (with ties) $\pi: \itms \to [|\itms|]$ if $\pi$ ranks the goods in decreasing order of $u$, i.e., $\pi(g) \le \pi(g')$ iff $u(g) \ge u(g')$ for all $g,g' \in \itms$. Here, a weaker notion is that of consistency: $u$ is said to be consistent with $\pi$ if $\pi(g) \le \pi(g')$ implies $u(g) \ge u(g')$ for all $g,g' \in \itms$. 

Note that SD-EF1 is a stronger notion than EF1: if an allocation is SD-EF1 with respect to utilities $\left( u_i \right)_{i \in \ags}$, then it must be EF1 with respect to $\left( u_i \right)_{i \in \ags}$, as well as with respect to any other utility profile $ \left( u'_i \right)_{i \in \ags}$ where, for each $i$, the utility $\av'_i$ is consistent with the ranking $\pi_i$ induced by $\av_i$.  

The following lemma provides a useful observation about SD-EF1 with respect to a profile with identical valuations; recall that the market valuation can be seen as such a profile.

\begin{lemma}
\label{lem:sdef1-constraints}
Let $\pi : \itms \leftrightarrow [m]$ be a strict ranking over $\itms$. If $u_i$ induces $\pi$ for every $i \in \ags$, then $\alc$ is SD-EF1 with respect to $\left( u_i \right)_{i \in \ags}$ iff, for all agents $i \in \ags$ and each index $\ell \in \{1, 2, \ldots, \lceil m/n\rceil \}$, we have $|\alc_i \cap \set{a \in \itms: n\ell - n + 1 \le \mpref(a) \le n\ell}| \le 1$. Further, if $u'_i$ is consistent with $\pi$ for every $i \in \ags$, then such an allocation is still SD-EF1 with respect to $\left( u'_i \right)_{i \in \ags}$.
\end{lemma}
\begin{proof}
For $\ell \in [\ceil{m/n}]$, let $C_{\ell} = \{\itm \in \itms: n\ell -n + 1 \le \pi(\itm) \le n\ell \}$. First, we prove the if and only if condition. 

($\Rightarrow$) Suppose this is not true. Let $k$ be the smallest value for which there exists an agent $j$ with $|A_j \cap C_k| > 1$, and since $|C_k| \le n$, there must exist another agent $i$ with $|A_i \cap C_k| = 0$. By the choice of $k$, for all $\ell < k$ and agent $i$, we must have $|A_i \cap C_\ell| \le 1$, and since $|C_\ell| = n$, this would imply $|A_i \cap C_\ell| = 1$. Now, since $|A_j \cap C_k| \ge 2$, consider any two goods $g,g' \in |A_j \cap C_k|$ and, without loss of generality, suppose $g \succ_\pi g'$. Taking $T = \set{g'' \in \itms : g'' \succ_\pi g'}$, we see that $|A_i \cap T| = k-1$ whereas $|A_j \cap T| = k+1$, violating the SD-EF1 property. 

($\Leftarrow$) Suppose $A$ is an allocation satisfying the given condition. Consider any two agents $i$ and $j$, and any good $g \in \itms$. Let $T = \set{g' \in \itms: g' \succ_\pi g}$. We want to prove that $|A_i \cap T| \ge |A_j \cap T|-1$. Suppose $g \in C_k$. Then, the desired condition follows from observing, as above, that $|A_i \cap C_{k'}| = |A_j \cap C_{k'}| = 1$ for all $k' < k$ (for which $|C_{k'}| = n$) and $|A_j \cap C_k| \le 1$. 

For the last part of the lemma, it suffices to observe that when $u'_i$ is consistent with $\pi$, for every good $g \in \itms$, there exists a good $g' \in \itms$ such that $\set{g'' \in \itms: u'_i(g'') \ge u'_i(g)} = \set{g'' \in \itms: g'' \succ_\pi g'}$. Hence, SD-EF1 with respect to $\pi$ imposes a strictly greater set of conditions than SD-EF1 with respect to $\avs'$.  
\end{proof}

\section{Simultaneous Envy-Freeness}\label{sec:ef-both-sides}
This section first shows that the strong notion of SD-EF1 cannot be guaranteed \wrt both the subjective utilities and the market valuation simultaneously. 

\begin{theorem}\label{thm:sd-ef1-both-impossible}
There exists an instance with $2$ agents and $7$ goods in which no allocation is SD-EF1 \wrt both the subjective utilities and the market valuation.
\end{theorem}
\begin{proof}
Let $\mpref = (g_1 \succ g_2 \succ g_3 \succ g_4 \succ g_5 \succ g_6 \succ g_7)$ be the ranking induced by the market values, and those induced by the subjective utilities of the two agents be
\begin{align*}
&\pref_1 = (g_1 \succ g_3 \succ g_2 \succ g_5 \succ g_4 \succ g_7 \succ g_6),\\
&\pref_{2} = (g_1 \succ g_5 \succ g_2 \succ g_3 \succ g_6 \succ g_7 \succ g_4).
\end{align*}
The impossibility result stated in the theorem will hold for any subjective utilities and market valuation that induce the above-mentioned rankings.

By \Cref{lem:sdef1-constraints},
SD-EF1 \wrt the market valuation demands that each agent receive exactly one good from each of the sets $C_1 = \{g_1, g_2\}$, $C_2 = \{g_3, g_4\}$, and $C_3 = \{g_5, g_6\}$.

Suppose agent $1$ receives good $g_1$. Then, agent $2$ must receive both good $g_5$ (for SD-EF1 \wrt the subjective utilities) and good $g_2$ (for SD-EF1 \wrt the market values). Then, considering the top four  goods in $\pref_1$, agent $1$ must receive $g_3$ to satisfy SD-EF1 \wrt its subjective utility. Next, recalling the constraints imposed by SD-EF1 \wrt market values through sets $C_2$ and $C_3$ above, we conclude that agent $2$ must get $g_4$ and agent $1$ must get $g_6$. So far, we have $\set{g_1,g_3,g_6} \subseteq \alc_1$ and $\set{g_2,g_4,g_5} \subseteq \alc_2$. Since both agents have only two of their top six goods, both agents would need to have $g_7$ in order to satisfy SD-EF1 \wrt their subjective utilities, which is a contradiction. 

The case of agent $2$ receiving good $g_1$ leads to a similar argument, where we can deduce $\set{g_2,g_3,g_6} \subseteq \alc_1$ and $\set{g_1,g_4,g_5} \subseteq \alc_2$. Hence, in this case as well, both agents would require $g_7$ for SD-EF1 \wrt their subjective utilities, which is a contradiction. 

\end{proof}

We complete the picture by providing positive results for achieving SD-EF1 \wrt both the subjective utilities and the market values in two special cases. 
\begin{theorem}\label{thm:sd-ef1-2a-6g}
In instances with two agents and at most six goods, there always exists an allocation that is SD-EF1 \wrt both the subjective utilities and the market valuation. 
\end{theorem}
\begin{proof}
We propose a cut-and-choose style algorithm. By invoking \Cref{thm:sd-ef1-both-identical} for $\mv$ as the market valuation and $\av_1$ as the identical utilities, we find a two-bundle partition of goods $(A_1, A_2)$ that satisfies SD-EF1 \wrt\ both $\av_1$ and $\mv$. Agent $1$ accepts both bundles $A_1$ and $A_2$ with no SD-EF1 violation. 

We will show that agent $2$ accepts at least one of the two bundles towards SD-EF1 \wrt utility $\av_2$. Write $\pi_2$ to denote the ranking induced by utility $\av_2$ and, for each index $\ell \in [m]$, let $\itms_{2, \ell}:=\{g \in \itms \mid \pi_2(g) \leq \ell \}$. That is, $\itms_{2, \ell}$ is the set of the top $\ell$ goods under agent $2$'s utility $u_2$. 

Also, write $e_{\ell} := |A_1 \cap \itms_{2, \ell}| - |A_2 \cap \itms_{2, \ell}|$ to capture the split of the top-$\ell$ goods  between the two subsets $A_1$ and $A_2$. Note that, if, for all indices $\ell \in [m]$, we have $e_{\ell} \ge -1$, then allocating $A_1$ to agent $2$ satisfies SD-EF1 for her. This follows from the fact that, in such a case, agent $2$ receives at most one item less than agent $1$ in each $\itms_{2,\ell}$. Similarly, if $e_{\ell} \le 1$ for all $\ell \in [m]$, allocating $A_2$ 
to agent $2$ satisfies SD-EF1 for her. Next, we show that at least one of the two conditions must hold. Suppose by contradiction that there exists an $\ell_1, \ell_2 \in [m]$ such that $e_{\ell_1} \ge 2$ and $e_{\ell_2} \le -2$. Since $|e_{\ell_1} - e_{\ell_2}| \ge 4$, either $|A_1| \ge 4$ or $|A_2| \ge 4$, which contradicts the fact that the bundle sizes are balanced, i.e., $\max\{|A_1|, |A_2|\} \le \lceil m / 2 \rceil \le 3$, where the last inequality  follows from the assumption that $m \le 6$. The theorem stands proved.  
\end{proof}

\begin{theorem}\label{thm:sd-ef1-both-identical}
When the subjective utilities of the agents induce the same ranking over the goods, an allocation that is SD-EF1 \wrt the subjective utilities and the market valuation always exists, and can be computed in polynomial time.
\end{theorem}
\begin{proof}
If $m$ is not divisible by $n$, add dummy goods such that the number of items becomes divisible by $n$. For each dummy good both market value and subjective utility for all the agents is equal to $0$. Now, by construction, we have set $\itms'$ of $m' = nk$ items. Considering this integer $k$, construct a bipartite graph $G$ with vertices $\set{w_1, \ldots, w_k}$ on the left side and $\set{z_1, \ldots, z_k}$ on the right side.

Let $\pi_v$ be the ranking induced the market valuation $\mv$. Also, let $\pi_u$ be a ranking consistent with all the subjective utilities $u_i$. By the theorem assumption, such a ranking $\pi_u$ exists. 

For any good $a \in \itms'$, let $r_a, r'_a \in [m']$ be the ranks of $a$ \wrt\ the two orderings, respectively, i.e., $\pi_v(r_a) = \pi_u(r'_a) = a$. Then, in the bipartite graph $G$, include an edge from $w_{\lceil r_a / n \rceil}$ to $z_{\lceil r'_a / n \rceil}$ with label $a$. Note that, each vertex in $G$ is of degree $n$. The $n$-regular bipartite graph $G$ can be decomposed into $n$ perfect matchings, by Hall's theorem. Let the bundle allocated to agent $i$ be the set of goods labeled on the edges of the $i$-th matching. Note that each item is used as a label of a single edge, so the bundles are disjoint. In addition, for every index $\ell \in [k]$, each bundle has exactly one item $a$, with $r_a \in \{n\ell -n + 1, \ldots, n\ell\}$ and exactly one item $b$, with $r'_b \in \{n\ell -n + 1, \ldots, n\ell\}$. After removing the dummy items we have $|\alc_i \cap \set{a \in \itms: n\ell - n + 1 \le r_a \le n\ell}| \le 1$, and similarly $|\alc_i \cap \set{a \in \itms: n\ell - n + 1 \le r'_a \le n\ell}| \le 1$ for $\ell \in [k]$. Therefore, by \Cref{lem:sdef1-constraints}, allocation $\alc$ is SD-EF1 \wrt both market values and subjective utilities.
\end{proof}
Next, let us consider relaxing the goal of SD-EF1 to (the regular) EF1, at least on one side. Doing so on the market side leads to the following open problem, which seems to require considerably novel ideas towards a resolution. 

\begin{open}
\label{open:ef1-sdef1}
    Does there always exists an allocation that is SD-EF1 \wrt the subjective utilities and EF1 \wrt the market valuation?
\end{open}

When we relax SD-EF1 to EF1 on the agents' side, a positive result emerges, even in the general case. We remark that \citet{bu2023fair} prove the existence of an allocation that is EF1 \wrt both the subjective utilities and the market values. Upon careful observation, however, one finds that not only do they actually achieve SD-EF1 \wrt the market values, but in fact, their algorithm is precisely an instantiation of a more general algorithm designed previously by \citet{BB18}. In other words, this problem can be reduced to the result of \citet{BB18}; we lay out the simple reduction here, highlighting that there is no need to derive this result from scratch.  
 
\citet{BB18} consider a setting in which there are no market values, but we are given a partition $\set{\itms_1, \ldots, \itms_\ell}$ of the goods along with cardinality bounds $h_1, \ldots, h_\ell \in \mathbb{Z}_+$. The goal here is to find a fair allocation $\alc$ that satisfies cardinality constraints: $|A_i \cap \itms_k| \le h_k$ for each $k \in [\ell]$ and agent $i \in \ags$. Note that, we need the bound $h_k \ge \lceil |\itms_k|/n \rceil$, for each $k \in [\ell]$, to even ensure the existence of an allocation which satisfies the cardinality constraints. \citet{BB18} show that, in such a case,  an EF1 allocation satisfying the cardinality constraints is guaranteed to exist. 

For our reduction, we consider $\pi_v$ the ranking induced by the market valuation $\mv$. Then, we set $\ell := \ceil{m/n}$, and, for each $k \in [\ell]$, define $\itms_k := \set{g \in \itms: (k-1) n < \pi_v(g) \le k n}$. That is, $\itms_1$ is the set of top $n$ goods according to the market values, $\itms_2$ is the set of next $n$ goods according to the market values, and so on. Finally, we impose the cardinality constraints $|A_i \cap \itms_k| \le 1$ for all $i \in \ags$ and $k \in [\ell]$. From \Cref{lem:sdef1-constraints}, we see that this implies SD-EF1 \wrt the market values. The polynomial-time algorithm of \citet{BB18} then immediately yields an allocation that is EF1 \wrt the subjective utilities while satisfying the cardinality constraints (and, hence, SD-EF1 \wrt the market values). 

\begin{corollary}
\label{cor:ef1-sdef1}
An allocation that is EF1 \wrt the subjective utilities and SD-EF1 \wrt the market valuation always exists, and can be computed in polynomial time.
\end{corollary}

\section{Extensions to Other Criteria}\label{sec:extensions}

Let us now extend our setup of fair division with market values to (i) consider other desiderata such as Pareto optimality, MMS, and EFX, and (ii) go beyond additive utilities. 

\subsection{Pareto Optimality}\label{sec:indiv-PO}

\begin{definition}[Pareto Optimality]
\label{def:po}
An allocation $A$ is said to be \emph{Pareto optimal} (PO) \wrt the subjective utilities if there is no allocation $A'$ with the property that $\av_i(A'_i) \ge \av_i(A_i)$, for all $i \in \ags$, with at least one inequality being strict. 
\end{definition}

\Cref{app:fpo} addresses a strengthening called \emph{fractional Pareto optimality} (fPO), which demands that not even a fractional allocation (where the goods may be fractionally assigned to agents) Pareto dominates the given integral one. 

First, we note that every allocation is trivially PO \wrt the market values. Hence, the only interesting consideration is to demand PO \wrt the subjective utilities. Here, we show that PO cannot be achieved  with SD-EF1 (Theorem \ref{thm:po-sdef1-impossible}). Interestingly, SD-EF1 implies both EF1 and balancedness ($|A_i| \in \set{\lfloor m/n \rfloor, \ceil{m/n}}$ for all $i$), and both these relaxations can individually be attained together with PO~\cite{CKMP+19,CFS17}. 

\begin{theorem}\label{thm:po-sdef1-impossible}
There exists an instance with 2 agents and 6 goods in which there is no allocation that is PO \wrt the subjective utilities and SD-EF1 \wrt the market valuation. 
\end{theorem}
\begin{proof}
Consider an instance with two agents with the following subjective utilities for goods $g_1,\ldots,g_6$:
\begin{table}[htb!]
\centering
\renewcommand{\arraystretch}{1.2}
\begin{tabular}{c|c c c c c c}
& $g_1$ & $g_2$ & $g_3$ & $g_4$ & $g_5$ & $g_6$\\
\hline
Agent $1$ & 19 & 7 & 4 & 3 & 2 & 1\\
Agent $2$ & 8 & 7 & 6 & 5 & 4 & 3\\
\hline
\end{tabular}
\end{table}

In addition, let the additive market valuation be identical to the utility of the first agent, i.e., $\mv=\av_1$. Note that both agents must receive three goods, respectively, to ensure SD-EF1 (across the six goods) \wrt the market valuation. 

Since agent $1$ would be happy to give up any three goods to receive $g_1$ and agent $2$ would be happy to accept any three goods in exchange for $g_1$, PO demands that $g_1$ must be given to agent $1$. Then, for SD-EF1, $g_2$ must be given to agent $2$. 

Among the remaining four goods $\set{g_3,g_4,g_5,g_6}$, each agent must receive two goods in such a way that no agent gets both $g_3$ and $g_4$ (or equivalently, both $g_5$ and $g_6$). However, it is easy to see that agent $1$ would be happy to give up $A_1 \cap \set{g_3,g_4,g_5,g_6}$ to receive $\set{g_2}$, and agent $2$ would also be happy to give up $\set{g_2}$ to receive $A_1 \cap \set{g_3,g_4,g_5,g_6}$, violating PO.
\end{proof}

This leads us, in the PO context, to the following possibility: EF1+PO \wrt the subjective utilities and EF1 \wrt the market values. We leave this as an open question.

\begin{open}
\label{open:ef1po-ef1}
    Does there always exists an allocation that is EF1+PO \wrt the subjective utilities and EF1 \wrt the market values (or even balanced)? 
\end{open}

Finally, relaxing SD-EF1 to EF1 \wrt the market values, we show that PO \wrt the subjective utilities can be attained together with EF1 \wrt the market values. Note that if the market values of all goods are set to be equal, then EF1 is equivalent to balancedness; thus, our result strictly strengthens the known result that a balanced PO allocation always exists~\citet{CFS17}. 

In fact, we prove a slightly more general result, whereby we allow the market values to heterogeneous (where each agent $i$ has an additive market valuation $\mv_i$) and achieve \emph{equitability up to one good} (EQ1)~\cite{freeman2019} with respect to it. 

\begin{definition}[EQ1]
    An allocation $\alc$ is called equitable up to one good (EQ1) \wrt a valuation profile $\mvs = \left( v_i \right)_{i \in \ags}$ if, for every pair of agents $i,j \in \ags$, either $\alc_j = \emptyset$ or $\mv_i(\alc_i) \ge \mv_j(\alc_j \setminus \set{g})$ for some good $g \in \alc_j$. 
\end{definition}

\citet{freeman2019} prove that an allocation satisfying EQ1+PO does not always exist, but it does (and even EQ1+fPO can be guaranteed) when all values are nonzero (i.e., strictly positive). To obtain EQ1 and PO \wrt two different valuation profiles, respectively, it is clear that the profile \wrt which we want PO must have nonzero values.\footnote{Otherwise, consider an example in which three goods are positively valued only by agent $1$, and a fourth good is positively valued only by agent $2$. The only PO allocation is to give the first three goods to agent $1$ and the fourth good to agent $2$. However, this is not EQ1 \wrt a different valuation profile where all agents value all goods at $1$.} We show that PO (even fPO) \wrt a nonzero valuation profile and EQ1 \wrt a different valuation profile can be attained, strictly strengthening the result of \citet{freeman2019}.

\begin{restatable}{theorem}{eqonepo}
\label{thm:eq1po}
An allocation that is fPO \wrt nonzero subjective utilities $(\av_i)_{i \in \ags}$ and EQ1 \wrt (heterogeneous) market valuations $\left(\mv_i \right)_{i \in \ags}$ always exists, and can be computed in time polynomial in $n$, $m$, and $V$, where $V = \max_{i \in \ags} |\set{\mv_i(S) : S \subseteq \itms}|$ is the maximum number of distinct value levels for any $i$.
\end{restatable}

Since both profiles are heterogeneous, one can swap them to achieve EQ1 \wrt subjective utilities and fPO \wrt heterogeneous nonzero market values as well. Since EQ1 and EF1 coincide when the market valuations  are identical, $\mv_i = \mv$ for  $i$, which is the case we consider throughout the paper, we obtain the following corollary.

\begin{corollary}
\label{cor:ef1-fpo}
There always exist an allocation that is fPO \wrt nonzero subjective utilities and EF1 \wrt the market valuation.
\end{corollary}

\subsection{MMS and EFX}\label{sec:mms-efx}
In addition to EF1, two fairness criteria that have received significant attention in the discrete fair division are \emph{maximin share fairness} (MMS) and \emph{envy-freeness up to any good} (EFX). Here, we show that these notions lead to stark impossibilities in our setup of fair division with market values. 

\begin{definition}[$\alpha$-MMS]
    Under utilities $\left( \av_i \right)_{i \in \ags}$, the maximin share of agent $i$ is defined as $\mu_i := \max_{(B_1, \ldots, B_n)} \ \min_{k \in [n]} \av_i(B_k)$; here the $\max$ is considered over all $n$-partitions of $\itms$. For $\alpha \in [0,1]$, an allocation $A$ is said to be $\alpha$-MMS, \wrt subjective utilities, if $\av_i(A_i) \ge \alpha \cdot \mu_i$ for all agents $i \in \ags$. 
\end{definition}

We know that EF1 implies $(\nicefrac{1}{n})$-MMS~\cite{CKMP+19}. Hence, in the guarantee of \Cref{cor:ef1-sdef1} (EF1 \wrt the subjective utilities and SD-EF1 \wrt the market values), the fairness property on one or both sides can be replaced by $(\nicefrac{1}{n})$-MMS. However, the following result shows that this is the best one can hope for, despite the fact that in the traditional fair division setup, even $\left(\frac{3}{4}+ \frac{1}{12n} \right)$-MMS is known to be exist~\cite{AG24}. 

\begin{theorem}\label{thm:mms}
    For any $\alpha > \nicefrac{1}{n}$, there exists an instance with two identical valuation profiles (either one can be viewed as consisting of identical subjective utilities and the other as consisting of market values) in which no allocation is simultaneously $\alpha$-MMS \wrt the first profile and either EF1 or $\alpha$-MMS \wrt the second, 
\end{theorem}
\begin{proof}
    Consider an instance with $n$ agents and $2n-1$ goods, where the goods are partitioned into two sets, $\itms_1$ with $n$ goods and $\itms_2$ with $n-1$ goods. Consider two identical valuation profiles with functions $\av$ and $\mv$ as follows:
    \begin{align*}
        & \av(\itm) = \begin{cases}
        1 & \itm \in \itms_1 \\
        n & \itm \in \itms_2
        \end{cases},
        &\mv(\itm) = \begin{cases}
        1 & \itm \in \itms_1 \\
        0 & \itm \in \itms_2
        \end{cases}.
    \end{align*}
    Note that EF1 or $\alpha$-MMS (for any $\alpha > 0$) \wrt $\mv$ requires that each agent receive exactly one good from $\itms_1$. Since $|\itms_2| < n$, there exists an agent $i \in \ags$ who does not receive any good from $\itms_2$. For this agent, $\av_i(\alc_i) = 1$ whereas, under $\av$, the maximin share $\mu_i = n$, causing a violation of $\alpha$-MMS \wrt $\av$, for any $\alpha > \nicefrac{1}{n}$.
\end{proof}

\begin{definition}[$\alpha$-EFX]
For $\alpha \in [0,1]$, allocation $\alc$ is said to be $\alpha$-envy-free up to any good ($\alpha$-EFX), \wrt the subjective utilities, if, for every pair of agents $i,j \in \ags$, either $\alc_j = \emptyset$ or $\av_i(\alc_i) \ge \alpha \cdot \av_i(\alc_j \setminus \set{g})$ for every $g \in \alc_j$.
\end{definition}

\begin{theorem}\label{thm:efx}
For any $\alpha \in (0,1]$, there exists an instance with two identical valuation profiles (either one can be viewed as consisting of identical subjective utilities and the other as market values) in which no allocation is simultaneously EF1 \wrt the first profile and $\alpha$-EFX \wrt the second.
\end{theorem}
\begin{proof}
    Consider an instance with $2$ agents, $4$ goods $\set{g_1,g_2,g_3,g_4}$, and two identical valuation profiles with valuation functions $\av$ and $\mv$ as follows: $\av(g_t) = 1$ for all $t \in [4]$, $\mv(g_1) = 2/\alpha+1$, and $\mv(g_t) = 1$ for $t \in \set{2,3,4}$. 

    Note that EF1 \wrt $\av$ requires that each agent receive two goods. Consider any allocation $\alc$. Without loss of generality, assume that agent $1$ receives $g_1$ and $g_2$. Then, $\mv(\alc_2) = 2$ whereas $\alpha \  \mv(\alc_1 \setminus \set{g_2}) = 2+\alpha > 2$, demonstrating a violation of $\alpha$-EFX.
\end{proof}

\begin{remark}
     Since EFX implies EF1, we can replace EF1 with EFX in \Cref{thm:mms} to deduce that EFX \wrt one profile and $\alpha$-MMS (for any $\alpha > \frac{1}{n}$) \wrt to the other cannot be guaranteed.
\end{remark}

\subsection{Monotone Valuations}\label{sec:monotone}
This section addresses whether one can extend the positive result of \Cref{cor:ef1-sdef1} to hold beyond additive valuations. That is, we consider whether EF1 can be achieved simultaneously for monotone (but not necessarily additive) subjective utilities and market valuation.\footnote{SD-EF1 is not well-defined beyond additive valuations: there is no natural ordering over individual goods induced by a monotone valuation, and even if one were to induce it somehow, it would not suffice to guarantee EF1 \wrt every monotone valuation inducing it. Hence, we limit our attention to EF1.} 

We are interested in the broad classes of monotone, subadditive valuations. A valuation function $f : 2^\itms \to \bbR$ is said to be monotone if $f(X) \le f(Y)$, for all  $X \subseteq Y \subseteq \itms$, and subadditive if $f(X \cup Y) \le f(X) + f(Y)$ for all $X,Y \subseteq \itms$. 

An allocation that is EF1 \wrt monotone subjective utilities is guaranteed to exist~\cite{LMMS04}. Notably, with market values, the question remains open. 

\begin{open}
    Does there always exist an allocation that is EF1 \wrt two heterogeneous monotone valuation profiles simultaneously (i.e., \wrt monotone subjective utilities and heterogeneous monotone market values)? What if one of the profiles consisted of additive valuations?
\end{open}

In fact, consider the special case of (homogeneous) additive market values where every good has the same market value. Then, EF1 \wrt such market values is simply balancedness, yielding the following tantalizing open question for the usual fair division setup without market values. 

\begin{open}
    Does there always exist a balanced EF1 allocation \wrt a (heterogeneous) monotone valuation profile?
\end{open}

While these questions remain open, we prove the existence of an allocation that is $\nicefrac{1}{2}$-EF1 \wrt subadditive subjective utilities and SD-EF1 \wrt additive market values. Formally, for $\alpha \in [0,1]$, an allocation $A$ is said to be $\alpha$-EF1 \wrt subjective utilities $\left( \av_i \right)_{i\in \ags}$ if, for each pair of agents $i,j \in \ags$ (with $A_j \neq \emptyset$), we have $\av_i(A_i) \geq \alpha \ \av_i(A_j \setminus \{g\})$ for some good $g \in A_j$. Here, we use the idea of {minimally envied subsets} from \cite{chaudhury2021little}. \\

\noindent
{\bf Minimal Envied Swap (MES) Algorithm.} As we did for establishing \Cref{cor:ef1-sdef1}, we consider $\pi_v$ the ranking induced by the (additive) market valuation $\mv$. Then, we partition the items into $(\itms_1,\ldots,\itms_{\ceil{m/n}})$ such that $\itms_k := \set{g \in \itms : (k-1)n < \pi_\mv(g) \le kn}$. As noted previously, for any allocation $A$, if $|A_i \cap \itms_k| \le 1$ for each $k \in [\ceil{m/n}]$ and all agents $i$, then $A$ is SD-EF1 \wrt $v$. We now provide an algorithm that finds such an allocation that is $\nicefrac{1}{2}$-EF1 \wrt monotone, subadditive subjective utilities.

We start with an allocation $A$ with empty bundles, $A_i = \emptyset$, for all agents $i \in \ags$. Also, initialize  `charity' $C = \itms$ as the set of all the unallocated goods. A subset $T \subseteq C$ is said to be an \emph{envied feasible} subset of the charity if (1) there exists an agent $i \in \ags$ envying it, i.e., $u_i(A_i) < u_i(T)$, and (2) $|T \cap \itms_k| \le 1$ for each $k \in [\ceil{m/n}]$. Further, we say that an envied feasible subset $T \subseteq P$ is \emph{minimal} if no agent envies any strict subset of $T$. It is easy to verify that while there exists an envied feasible subset of the charity, there also exists a minimal envied feasible subset of the charity. The algorithm now works in two phases.

\noindent\textbf{Phase 1:} While there exists a minimal envied feasible subset $T$ of the charity $C$, select an agent $i$ who envies $T$, and update $A_i = T$ and charity (unallocated good) $C =  \itms \setminus \left(\cup_{i\in \ags} A_i \right)$.
    
\noindent\textbf{Phase 2:} Once Phase 1 has terminated, allocate goods in $C$ arbitrarily subject to maintaining $|A_i \cap \itms_k| \le 1$ for all agents $i \in \ags$ and each $k \in [\ceil{m/n}]$. 

\begin{theorem}
\label{theorem:mono-half-EF1}
Every fair division instance with (monotone) subadditive subjective utilities and additive market values admits an allocation that is $\nicefrac{1}{2}$-EF1 \wrt subjective utilities and SD-EF1 \wrt market valuation. Further, such an allocation is returned by the MES algorithm detailed above. 
\end{theorem}
\begin{proof}
First, we argue that the algorithm terminates and returns an allocation that is SD-EF1 \wrt the market values. Note that each iteration of Phase 1 strictly increases the utility of the selected agent, and the utilities of the remaining agents remain unchanged. Hence, Phase 1 must terminate. Further, throughout the execution of Phase 1, the selected subset of goods $T$ is feasible, i.e., throughout, we maintain $|A_i \cap \itms_k| \le 1$ for each $i \in \ags$ and $k \in [\ceil{m/n}]$. Therefore, at the end of Phase 1 and for each $k$, the number of agents with exactly one good from $\itms_k$ is equal to $|\itms_k \setminus C|$. This observation implies that there remain $n - |\itms_k \setminus C|$ agents who can each still receive a good from $\itms_k$, while maintaining the cardinality constraint. Since $|\itms_k| \leq n$, we obtain the following bound for the unassigned goods $|\itms_k \cap C| \leq n - |\itms_k \setminus C|$. Hence, we can allocate the unassigned goods in $ \itms_k \cap C$, one each, to the remaining agents while maintaining the feasibility constraints. As noted above, this ensures SD-EF1 \wrt the market values.  

To show that the returned allocation is $\nicefrac{1}{2}$-EF1 \wrt the utilities, we note that the initial allocation (with empty bundles) is EF1. Further, during Phase 1, when we assign $A_i = T$, no agent envies any strict subset of $A_i$. Therefore, the allocation remains EF1 even after the swap and, hence, the EF1 property is maintained till the end of Phase 1. During Phase 2, each agent $i$ receives a subset of goods $B_i \subseteq C$ that is feasible ($|B_i \cap \itms_k|\leq 1$ for all $k$) and, hence, is not envied by any other agent $j$. This follows from the fact that there are no envied feasible subsets of the charity $C$ left after Phase 1. Finally, the subadditivity of the utilities ensures that the final allocation is $\nicefrac{1}{2}$-EF1. 
The theorem stands proved. 
\end{proof}
Next, we provide an exact version of Theorem \ref{theorem:mono-half-EF1} for the special case of two agents. The proof of this theorem closely follows from a result of \citet{KSV20}. Specifically, in Corollary 5.3, they show that for two agents with (arbitrary) monotone subjective utilities, there always exists a balanced EF1 allocation. In particular, their algorithm arbitrarily pairs up the goods and ensures that in the final EF1 allocation the two goods in each pair are given to different agents. We utilize this result to obtain the corollary below. Write $\pi_\mv$ to denote the ranking induced by the additive market valuation $\mv$ and pair the goods as $(\pi_v(1), \pi_v(2)), (\pi_v(3), \pi_v(4)), \ldots, (\pi_v(m-1), \pi_v(m))$.\footnote{If $m$ is odd, we can add a dummy good at the end, and let both agents have zero marginal subjective utility for it. Removing this good in the end preserves both our desired properties. This is how \citet{KSV20} handle an odd number of goods as well.} Note that if the two agents, respectively, receive exactly one good from every pair, then we get SD-EF1 \wrt market values. This reduction yields the result below. One can also verify that the algorithm of \citet{KSV20} can be simulated with polynomially many value queries to the subjective utilities (where, in a single query, one asks for the utility of an agent for a bundle of goods).

\begin{theorem}
\label{theorem:mono-two-agents}
    For the case of two agents, there always exists an allocation that is EF1 \wrt monotone subjective utilities and SD-EF1 \wrt additive market values, and it can be computed using polynomially many value queries. 
\end{theorem}

\section{Cake Division}\label{sec:cake}
This section extends our fair division with market valuation setup to heterogeneous divisible goods, i.e., to the canonical setting of cake cutting. First, we briefly introduce the cake-cutting model. The cake is represented by the interval $[0, 1]$. A piece of the cake is a Borel subset of $[0,1]$. Valuations over the cake are defined via density functions. Specifically, under an integrable density function $f : [0,1] \to \R$, the value assigned to any piece $X$ of the cake is denoted (with a slight abuse of notation) as $f(X) := \int_X f(x) \ \dd x$; this induces a (countably additive) measure over the cake that is absolutely continuous with respect to the Lebesgue measure and, hence, non-atomic.\footnote{Integrating a density function is equivalent to assuming the measure to be absolutely continuous with respect to the Lebesgue measure, while non-atomicity of the measure is a weaker condition; see \cite{SS16,SS19}.} 

Let $\ags$ denote a set of $n$ agents participating in the cake division. An allocation $A = (A_1,\ldots,A_n)$ is a partition of the cake $[0,1]$ into $n$ pairwise-disjoint\footnote{The pieces are allowed to have an overlap of Lebesgue measure zero.} pieces. Each agent $i \in \ags$ has a subjective density function $u_i$; her utility in allocation $A=(A_1, \ldots, A_n)$ is $u_i(A_i)$. Also, there is a market density function $v$ that captures the market values of the pieces, i.e., $v(X)$ denotes the market value of piece $X$. 

An allocation $A=(A_1, \ldots, A_n)$ is envy-free (EF) \wrt subjective utilities if  $u_i(A_i) \geq u_i(A_j)$ for all agents $i,j \in \ags$, and envy-free with respect to the market valuation if $v(A_i) = v(A_j)$ for all $i,j \in \ags$. 

\subsection{Envy-Free Cake Division}\label{sec:cake-ef}

In this section, we study the existence and properties of cake allocations that are envy-free both \wrt subjective utilities and market valuation. The existence of such an allocation is rather easy to establish (as we note below), hence, we focus on two questions: (a) how many cuts are required in such an allocation, and (b) how many queries must be made in the Robertson-Webb model to find such an allocation.

\subsubsection{Number of Cuts}\label{sec:cake-cuts}
This subsection obtains upper and lower bounds on the number of cuts (i.e., the total number of intervals across all the pieces) required for the desired envy-free allocations. First, we provide a lower bound: There exist cake division instances wherein at least $2n-2$ cuts are required to achieve simultaneous envy-freeness.  

\begin{theorem}\label{thm:cake-lower}
    There exists cake division instances wherein any allocation that is envy-free---both \wrt the subjective utilities and the market valuation---has at least $2n-2$ cuts. This lower bound holds even when the agents' utilities are identical.
\end{theorem}
\begin{proof}
    Consider an instance in which the agents' (identical)  density function $u$ is $2$ over $[0,1/2]$ and $0$ over $[1/2,1]$. Also, the market density function $v$ is $1$ over the entire cake $[0,1]$. 

    Consider any allocation $A = (A_1,\ldots,A_n)$ that is EF with respect to both $u$ and $v$. For $A$ to be EF with respect to $u$ (for the $n$ agents), it must be the case that $[0,1/2]$ is divided equally, by length, between the $n$ pieces $A_1, \ldots, A_n$. Such a division requires at least $n-1$ cuts in $[0,1/2]$. 
    
    Analogously, for EF under $v$, the entire cake $[0,1]$ must also be divided equally---by length---between the $n$ pieces. Now, given that $[0,1/2]$ is already divided equally, the same also holds for $[1/2,1]$. This requires at least $n-1$ cuts in $[1/2,1]$ as well. Hence, in total, at least $2n-2$ cuts are necessary to obtain EF allocation $A$. The theorem stands proved. 
\end{proof}

Complementing the lower bound in \Cref{thm:cake-lower}, we next show that envy-freeness for $n$ subjective utilities and a market valuation can be achieved simultaneously with $n(n-1)$ cuts. We obtain this upper bound by invoking the following (necklace splitting)  result of \citet{alon1987splitting} that guarantees existence of {\it perfect} cake divisions. 

\begin{theorem}[\cite{alon1987splitting}]\label{thm:necklace}
Given any $t$ density functions $\mu_1,\ldots,\mu_t$, there exists an $n$-division $A = (A_1,\ldots,A_n)$ of the cake with at most $t(n-1)$ cuts and the property that $\mu_i(A_k) = \nicefrac{1}{n}$,  for all $i \in [t]$ and each $k \in [n]$. That is, $A$ provides an equal division under all the $t$ densities. 
\end{theorem}

Note that \Cref{thm:necklace} holds for any number of densities $t$. The theorem leads to the following fairness guarantee. 

\begin{theorem}
\label{thm:cake-upper}
In any cake division instance with $n$ agents, there always exists a cake division that has at most $n(n-1)$ cuts and is envy-free, under both the subjective utilities and the market valuation. Further, if the agents have identical utilities, then $2n-2$ cuts suffice for inducing simultaneous envy-freeness. 
\end{theorem}
\begin{proof}
For a cake division instance wherein the $n$ agents have subjective utilities $u_1, \ldots,u_n$, we invoke \Cref{thm:necklace} with $t=n$. In particular, we set the first $t-1$ ($=n-1$) densities as follows: $\mu_i = u_i$, for $1 \leq i \leq t-1$, and set the last density $\mu_t = v$. The resulting allocation $A$ will have $t(n-1) = n(n-1)$ cuts. Then, we let the $n$th agent select her most preferred piece in $A$, and assign the remaining pieces arbitrarily among the first $n-1$ agents. Index the pieces of $A$ such that $A_n$ denotes the piece selected by agent $n$ and $A_i$ is the piece assigned to agent $i \in [n-1]$.

The allocation is envy-free for agent $n$, since she selects her most preferred piece. Further, the equal-divisions property from \Cref{thm:necklace} ensures that, for each agent $i \in [n-1]$ and all $j \in [n]$, we have $u_i(A_i) = \mu_i(A_i) = \nicefrac{1}{n} = \mu_i(A_j) = u_i(A_j)$. That is, allocation $A$ is envy-free for the agents. The equal-divisions property for $v = \mu_n$ gives us envy-freeness under the market valuation. This establishes the first part of the theorem.  

For the second part, we apply \Cref{thm:necklace} with $t=2$ densities. In particular, we set $\mu_1 = u$, where $u$ is the identical utility across the agents, and $\mu_2 = v$. Here, the number of cuts is at most $t(n-1) = 2n-2$. Note that equal division under $\mu_1 = u$ ensures EF for the agents. Also, equal division under $\mu_2 = v$ ensures EF with respect to the market values. This completes the proof. 
\end{proof}

The gap between \Cref{thm:cake-lower} and \Cref{thm:cake-upper} leaves open the following tantalizing question.

\begin{open}
   Does there always exists an allocation of the cake with $O(n)$ cuts that is envy-free \wrt the subjective utilities and the market valuation?
\end{open}
We believe that the flexibility of only having to achieve EF with respect to the subjective utilities (rather than an equal division) should permit using $o(n^2)$ cuts, perhaps even just $\Theta(n)$ cuts. 

\subsubsection{Number of Queries}\label{sec:cake-queries}
Recall that, in the standard Robertson-Webb query model for cake division, the work of \citet{aziz2016discrete} provides a protocol that finds envy-free division---under subjective utilities---with a finite (albeit hyper-exponential) number of queries. We note that, by contrast, a protocol with finite query complexity does not exist when envy-freeness is additionally required \wrt market valuation. 

We obtain this impossibility result by invoking the lower bound of \citet{procaccia2017lower} which shows that an \emph{equitable} cake division, even among two agents, cannot be computed using finitely many queries. Specifically, a cake division $(X_1,\ldots, X_n)$ is said to be equitable \wrt densities $\mu_1, \ldots, \mu_n$ if $\mu_i(X_i) = \mu_j(X_j)$ for all $i$ and $j$.   

One can reduce equitable cake division---under two densities $\mu_1$ and $\mu_2$---to fair division with market valuation as follows. Consider a cake-division instance with two agents whose identical subjective utilities are $\mu_1$. Also, let the market valuation be $\mu_2$. Here, any division $A=(A_1, A_2)$ that is EF \wrt the subjective utilities $\mu_1$ and the market valuation $\mu_2$ is equitable, $\mu_1(A_1) = \mu_1(A_2) = 1/2$ and $\mu_2(A_2) = \mu_2(A_1) = 1/2$. Hence, using the lower bound from \cite{procaccia2017lower}, we obtain that, with finitely many queries, one cannot find a division that is EF simultaneously \wrt subjective utilities and market valuation.

\subsection{Pareto Optimality}\label{sec:cake-po}
This section addresses the inclusion of Pareto Efficiency as an additional desideratum in cake division. It is known that, under agents' subjective utilities, there always exists an cake allocation that is envy-free and Pareto optimal (PO)~\cite{Wel85}. In particular, an allocation that maximizes the Nash social welfare (i.e., the product of agents' utilities) is one such fair and efficient allocation~\cite{SS19}. Now, when considering the market valuation, note that any allocation will be trivially PO. 

Hence, an interesting question is whether there always exists a cake division that is EF and PO on the subjective side as well as EF on the market side. It is easy to observe that, as in the case of indivisible goods (\Cref{sec:indiv-PO}), such a guarantee cannot be achieved if the agents' densities are zero for parts of the cake. For example, consider an instance with two agents such that agent $1$ has positive density over $[0,2/3]$, but zero density over $[2/3,1]$. By contrast, agent $2$ has zero density in $[0,2/3]$, but positive value density over $[2/3,1]$. Note that the only allocation $A=(A_1, A_2)$ that is PO with respect to these agents' utilities entails $A_1 = [0,2/3]$  and $A_2 = [2/3,1]$. However, if the market density $v$ is uniform over the entire cake $[0,1]$, then EF for $v$ would demand that the pieces assigned to the two agents have equal length. Since this does not hold for the PO allocation $A$, we get that the instance at hand does not admit an allocation that is EF and PO on the subjective side, and EF on the market side.

In fact, even if we assume that the agents' subjective densities are strictly positive over the entire cake, the existence of the desired allocation is not guaranteed. To establish this negative result, we adapt a construction of \citet{BJK13}, who used it to show that an allocation of the cake that is EF, PO, and equitable (EQ) under agents' utilities does not always exist. Recall that an allocation $A = (A_1, \ldots, A_n)$ is said to be equitable if $u_i(A_i) = u_j(A_j)$ for all agents $i,j$. \citet{BJK13} proved the impossibility for $n \ge 3$ agents and we do the same. In particular, this negative negative result even rules out a balanced EF and PO cake allocation; balancedness, the requirement that the piece allocated to each agent has the same length, is equivalent to EF \wrt the uniform market valuation, and is therefore a weaker requirement than EF \wrt an arbitrary market valuation.

Notably, the requirements in \cite{BB14} (i.e., EF, PO, and EQ) can be achieved together for two agents. We leave it as an open question whether, in the current context, a positive result can be obtained specifically for two agents. 

\begin{open}
    Does there always exist an allocation of the cake between $n=2$ agents that is EF and PO \wrt the subjective values and EF \wrt the market value?
\end{open}

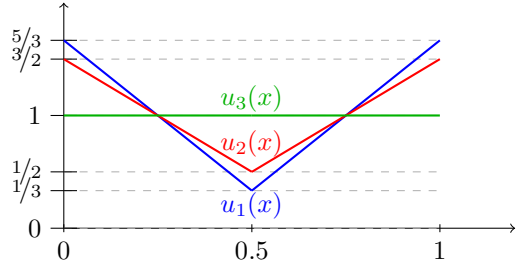
\begin{figure}[H]
    \centering
    \begin{tikzpicture}[xscale=5,yscale=1.5]
      \draw[->] (0,0) -- (1.2,0);
      \draw[->] (0,0) -- (0,2);
      \foreach \x/\xtext in {0/0, 0.5/0.5, 1/1}
        \draw (\x,1pt) -- (\x,-1pt) node[below] {$\xtext$};
      \foreach \y/\ytext in {0/0, {1/3}/\nicefrac{1}{3}, {1/2}/\nicefrac{1}{2}, 1/1, 1.5/\nicefrac{3}{2}, {5/3}/\nicefrac{5}{3}}{
        \draw[dashed, black!30!white] (0,\y) -- (1,\y);
        \draw (1pt,\y) -- (-1pt,\y) node[left] {$\ytext$};
      }
      \draw[blue, thick, domain=0:0.5] plot (\x, {5/3 - 8*\x/3});
      \draw[blue, thick, domain=0.5:1] plot (\x, {5/3 - 8*(1-\x)/3});
      \draw[red, thick, domain=0:0.5] plot (\x, {3/2 - 2*\x});
      \draw[red, thick, domain=0.5:1] plot (\x, {3/2 - 2*(1-\x)});
      \draw[green!70!black, thick, domain=0:0.5] plot (\x, {1});
      \draw[green!70!black, thick, domain=0.5:1] plot (\x, {1});
      \node[blue] at (0.5, 0.2) {$u_1(x)$};
      \node[red] at (0.5, 0.75) {$u_2(x)$};
      \node[green!70!black] at (0.5, 1.15) {$u_3(x)$};
    \end{tikzpicture}
    \caption{An instance in which no balanced EF+PO allocation exists.}
    \label{fig:balanced-ef-po}
\end{figure}

\begin{theorem}
\label{theorem:no-bal-ef-po}
    There exists a cake-division instance, with $3$ agents, wherein  no allocation is simultaneously balanced, EF, and PO.  Consequently, EF and PO \wrt the subjective utilities and EF \wrt the market valuation cannot be guaranteed in general.
\end{theorem}
\begin{proof}
    Consider an instance with three agents and subjective density functions, $u_1,u_2,u_3$, defined as follows (see also \Cref{fig:balanced-ef-po}): 
    \begin{align*}
    &u_1(x) = \begin{cases}
        \frac{5}{3}-\frac{8}{3} x & \text{ for } x \in [0,0.5],\\
        \frac{5}{3}-\frac{8}{3} (1-x) & \text{ for } x \in [0.5,1].
    \end{cases} \qquad \\
    & u_2(x) = \begin{cases}
        \frac{3}{2}-2x & \text{ for } x \in [0,0.5],\\
        \frac{3}{2}-2(1-x) & \text{ for } x \in [0.5,1].
    \end{cases}\\
    &u_3(x) = 1 \ \text{ for } x \in [0,1]   
    \end{align*}
    The market density function $v$ is uniform over the whole cake. Hence, EF \wrt the market values corresponds to balancedness (i.e., the pieces of the three agents must have equal Lebesgue measure). 

    First, we narrow down the space of candidate allocations using a result of \citet{barbanel1997two} that an allocation $A$ of the cake is PO iff it maximizes some weighted welfare, i.e., $A \in \argmax_B \sum_{i=1}^3 w_i \cdot u_i(B_i)$ for some $w_1,w_2,w_3 > 0$. Further, note that a weighted welfare is maximized iff  each portion of the cake is allocated to an agent with the highest weighted density for it. 
    
    This observation and the symmetry (around $1/2$) in the agents' utilities imply that, in any PO allocation, agent $3$ must be assigned the piece $[x^*,1-x^*]$ for some $x^* \in [0,1/2]$. Since the desired allocation has to be balanced, we obtain that agent $3$ must receive the piece $[1/3,2/3]$. Further, the weighted densities of agents $1$ and $2$ must intersect at $1/6$ and $5/6$, so that one of those two agents receives $[0,1/6] \cup [5/6,1]$ while the other receives $[1/6,1/3] \cup [2/3,5/6]$. Since both agents have higher utility for $[0,1/6] \cup [5/6,1]$, the agent receiving $[1/6,1/3] \cup [2/3,5/6]$ envies the other.\footnote{Further inspection tells us that the only balanced PO allocation is where $A_1 = [0,1/6] \cup [5/6,1]$, $A_2 = [1/6,1/3] \cup [2/3,5/6]$, and $A_3 = [1/3,2/3]$, and this maximizes weighted welfare where $w_1 : w_2 : w_3 = 126/110 : 6/5 : 1$.}
\end{proof}

If one wishes to forgo EF \wrt the subjective utilities, we show that there exists an allocation that is $\epsilon$-PO \wrt the subjective utilities and EF \wrt the market valuations, when the subjective density functions are strictly positive. This is an analogue of \Cref{thm:eq1po} in the cake-cutting model. In fact, we prove a more general result and allow \emph{heterogeneous} market valuations, whereby each agent $\ag$ has a (possibly different) market density function $\mv_\ag$.

\begin{theorem}
\label{thm:cake-eq-po-lipschitz}
For a cake-division instance with strictly positive subjective utility functions that are Lipschitz continuous and heterogeneous market values, there always exists an allocation of the cake that is $\epsilon$-PO \wrt the subjective utilities and EQ \wrt the market values. 
\end{theorem}

An allocation $\alc$ is equitable (EQ) with respect to market valuations if $\mv_{\ag}(\alc_{\ag}) = \mv_{\ag'}(\alc_{\ag'})$ for all agents $\ag, \ag' \in \ags$. A function $f: [0, 1] \to \bbR$ is $K$-Lipschitz if $|f(x) - f(y)| \le K |x - y|$ for all $x, y \in [0, 1]$. Next, we define $\epsilon$-Pareto optimality.

\begin{definition}[$\epsilon$-Pareto optimality]
For a given $\epsilon > 0$, an allocation $\alc$ is $\epsilon$-Pareto optimal \wrt density functions $\{u_\ag\}_{\ag \in \ags}$  if there does not exist another allocation $\alc'$ for which $u_\ag(\alc'_\ag) \ge u_{\ag}(\alc_{\ag}) + \epsilon$ for all $\ag \in \ags$.
\end{definition}

We use the Lipschitz continuity assumption in \Cref{thm:cake-eq-po-lipschitz} to approximate the subjective utilities with piecewise constant density function as closely as needed, as is similarly done in \citet{CLPP11} and \citet{RP98}. A piecewise constant function $h$ can be expressed with a union of finitely many intervals such that $h$ is a constant function over each interval.

\begin{lemma}[\citet{CLPP11,RP98}]
\label{lem:apx-lipschitz}
Given a real number $\epsilon > 0$ and a $K$-Lipschitz density function $v: [0, 1] \to \bbR$, there exists a piecewise constant function $h$ such that $v(x) \le h(x) < v(x) + \epsilon / 2$ for all $x \in [0, 1]$.
\end{lemma}
\begin{proof}
Partition the cake $[0, 1]$ into $\lceil 4K / \epsilon\rceil$ many intervals $\calI$ of length at most ${\epsilon}/{4K}$. For each interval $I \in \calI$, let $h(x) = \sup_{x \in [l, r]} v(x)$ for $x \in I$. This way, $f(x) \le h(x)$ and by the Lipschitz continuity and that the intervals are of length at most ${\epsilon}/{4K}$  we have $f(x) \le v(x) + \epsilon / 4 < v(x) + \epsilon / 2$.
\end{proof}

To prove \Cref{thm:cake-eq-po-lipschitz}, we round the subjective density functions as in \Cref{lem:apx-lipschitz}, and use the following result to achieve (exact) Pareto optimality \wrt the rounded subjective utilities and EQ \wrt market values. This is the crux of the whole result, and its detailed proof is presented in \Cref{app:cake-po}. 

\begin{theorem}\label{thm:cake-PO-EQ-infinitesimal}
For a cake division instance with strictly positive \emph{piecewise constant} subjective density functions $\{u_\ag\}_{\ag \in \ags}$ and heterogeneous market valuations  $\{v_i\}_{i \in \ags}$, there always exist an allocation of the cake that is PO \wrt the subjective utilities and EQ \wrt the market values.
\end{theorem}

The proof of \Cref{thm:cake-eq-po-lipschitz} is simply a corollary of \Cref{thm:cake-PO-EQ-infinitesimal} and \Cref{lem:apx-lipschitz}.

\begin{proof}[Proof of \Cref{thm:cake-eq-po-lipschitz}]
We approximate the subjective utilities as in \Cref{lem:apx-lipschitz} and obtain $\avs' = \{\av'_\ag\}_{\ag \in \ags}$. By invoking \Cref{thm:cake-PO-EQ-infinitesimal} with $\avs \gets \avs'$ and $\mv_i \gets \mv$, we get an allocation $\alc$ which is PO \wrt $\avs'$ and EQ \wrt $\mv$. Since EQ and EF coincide when the market valuations are identical, $\alc$ is EF \wrt $\mv$. Further, from the approximation guarantee of \Cref{lem:apx-lipschitz} and that $\alc$ is PO \wrt $\avs'$, we have $\alc$ is $\epsilon$-PO \wrt $\avs$.
\end{proof}

Similar to the discussion of \citet{CLPP11}, our results can be easily extended to cases where the utility functions have finitely many discontinuity points and are Lipschitz continuous on each interval induced by the discontinuities.

\section{Conclusion and Future Work}\label{sec:discussion}
The current work identifies a conceptually important frontier of fair division that requires fairness under market values, along with the standard desiderata under subjective utilities. This framework raises many interesting questions and admits technically rich connections with settings, such as fair division with constraints. Several of the raised questions are, in fact, fundamental even in the standard fair division setting (i.e., without an explicit invocation of the market values), e.g., the open question regarding the existence of balanced EF1 allocations under monotone utilities.  

For multiple criteria---including SD-EF1, EF1, PO, MMS, and EFX---we obtain positive and negative results on simultaneous fairness under subjective utilities and the market valuation. With this groundwork, interesting directions for future work in the current framework include the allocation of chores, division of homogeneous goods, and price of fairness. The open questions explicitly identified throughout the paper are also relevant avenues for exploration.

\bibliographystyle{abbrvnat}
\bibliography{abb,ultimate,bib}

\cleardoublepage

\onecolumn
\appendix
\section*{\LARGE \centering Appendix}

\section{Fractional Pareto Optimality}\label{app:fpo}

A fractional allocation $\alc \in [0, 1]^{\ags \times \itms}$ is a complete division of goods such that $\sum_{\ag \in \ags} \alc_{\ag, \itm} = 1$ for each $\itm \in \itms$. Next, we define fractional Pareto optimality which implies \Cref{def:po}.

\begin{definition}[Fractional Pareto Optimality] An allocation $\alc$ is \emph{fractionally} Pareto optimal (fPO) if there does not exist a \emph{fractional} allocation $\alc' \in [0, 1]^{\ags \times \itms}$ such that $\av_\ag(\alc'_\ag) \ge \av_\ag(\alc_\ag)$ with one inequality being strict.
\end{definition}

Similar to the discussion in \Cref{sec:indiv-PO}, every allocation is PO \wrt the market values and the interesting case is to consider PO \wrt subjective utilities. The impossibility of \Cref{thm:po-sdef1-impossible} extends to the stricter requirement of fPO \wrt the subjective utilities. However, we answer  a stronger version of \Cref{open:ef1po-ef1} in negative.  

\begin{proposition}\label{prop:ef1-fpo-nonexample}
There exists an instance with 4 goods and 2 agents where no allocation is EF1 and fPO \wrt subjective utilities and EF1 \wrt market values.
\end{proposition}

\begin{proof}
Consider an instance with two agents, four goods, and the following valuations.
\begin{itemize}
    \item Market valuation: $\mv(\itm) = 1, \forall \itm \in \itms$. 
    \item The first agent: $u_1: \{\itm_1: 4, \itm_2: 4, \itm_3: 1, \itm_4: 1\}$.
    \item The second agent: $u_2: \{\itm_1: 3, \itm_2: 3, \itm_3: 1, \itm_4: 1\}$.
\end{itemize}
The unique EF1 and fPO allocation \wrt subjective utilities is $\alc_1 = \{1\}$ and $\alc_{2} = \{2, 3, 4\}$. However, this is not EF1 \wrt the market values as agent $1$ envies agent $2$.
\end{proof}

\section{PO \wrt Subjective Utilities and EF1 \wrt Market Values}

In this section, we prove \Cref{thm:eq1po}. Our algorithm and analysis builds on the work of \citet{freeman2019} and \citet{murhekar2021fair} who show EQ1 and fPO allocations exist when all agents have strictly positive utilities. They build on the work of \citet{BKV18} who designed a pseudo-polynomial time algorithm to find an allocation of goods that is EF1 and PO using Fisher markets. At a high-level, the algorithm performs a local search over Pareto optimal allocations that eventually reaches an EF1 (and PO) allocation. 

We first recall Fisher markets and make the necessary definitions to describe the algorithm, and then outline the algorithm and how the proof of \citet{murhekar2021fair} and \citet{freeman2019} can be adapted to prove our \Cref{thm:eq1po}. 

\subsection{Fisher Markets}
\label{sec:fisher-markets}
In a \emph{Fisher market}, in addition to the utilities, we have a price vector $\pp = (p_1, \ldots, p_m)$ where $p_\itm$ is the price of good $\itm$. The \emph{bang-per-buck} ratio of good $\itm$ for agent $i$ is defined as $\bb_\ag(\itm) = {u_\ag(\itm)}/{p_\itm}$. We refer to the maximum of such ratios as the \emph{maximum bang per buck} (MBB) ratio of agent $\ag$ denoted by $\mbb_\ag = \max_{\itm \in \itms} \bb_\ag(\itm)$. Next, we define the Fisher market equilibrium.

\begin{definition}[Fisher Market Equilibrium]
An allocation $\alc$ and price vector $\pp$ is a \emph{Fisher market} equilibrium if all goods are allocated and each agent receives goods that have the maximum bang-per-buck ratio, i.e., $\bb_\ag(\itm) = \mbb_i$ for all agents $\ag$ and item $\itm \in \alc_\ag$.
\end{definition}

Every Fisher market equilibrium is Pareto optimal due to the first welfare theorem.
\begin{proposition}
Every allocation of a Fisher market equilibrium is Pareto optimal.
\end{proposition}

The following structures are useful for describing the algorithm. 

\begin{definition}[MBB Alternating Path]
For a Fisher market equilibrium $(\alc, \pp)$, we write $\ag_1 \stackrel{\itm}{\gets} i_2$ if agent $i_2$ is allocated good $\itm$, i.e., $\itm \in \alc_{i_2}$, and  good $\itm$ is an MBB good for agent $i_2$, i.e., $\bb_{i_1}(\itm) = \mbb_{i_1}$. 
An \emph{MBB alternating path} of length $\ell$ from $i_\ell$ to $i_0$ is a sequence $i_0 \stackrel{\itm_1}{\gets} i_1 \stackrel{\itm_2}{\gets} i_2 \stackrel{\itm_3}{\gets} \cdots \stackrel{\itm_\ell}{\gets} i_\ell$.
\end{definition}

\begin{definition}[MBB Hierarchy]
For a set of agents $\ags' \subseteq \ags$, the MBB hierarchy of $\ags'$ is the set of agents $\ag \in $ reachable from $\ags'$ through an MBB alternating path.
\end{definition}

\begin{algorithm}[t]
\DontPrintSemicolon
\SetAlgoNoEnd
\SetKwProg{PhaseI}{Phase 1}{}{}
\SetKwProg{PhaseII}{Phase 2}{}{}
\SetKwProg{PhaseIII}{Phase 3}{}{}
\caption{EQ1 \wrt Heterogeneous Market Values and PO \wrt Subjective Utilities}
\label{alg:eq1-po-two-valuations}
\KwIn{A fair division instance $\inst = (\ags, \itms, \mvs, \avs)$ with strictly positive agent utilities and non-negative heterogeneous market values}
\KwOut{An allocation $\alc$ that is EQ1 \wrt Heterogeneous Market Values and PO \wrt Subjective Utilities}
\PhaseI{Initialize by an equilibrium allocation \wrt subjective utilities, e.g., a social welfare maximizing allocation}{
$(\alc, \pp) \gets$ initial welfare maximizing allocation (\wrt subjective utilities) where $p_\itm = u_{i, \itm}$ for $\itm \in \alc_i$\;
}
\PhaseII{Reallocate goods towards a more equitable allocation \wrt market values}{
$L \gets \{i \in \ags \mid i \in \argmin_{h \in \ags} v_h(\alc_h)\}$\;
\If{$\exists i \in L, \exists$ MBB alternating path $i_0 \stackrel{\itm_1}\gets i_1 \stackrel{\itm_2}{\gets} \cdots \stackrel{\itm_\ell}{\gets} i_{\ell}$ with $v_{i_\ell}(\alc_{i_\ell} \setminus \{\itm_{\ell}\}) > v_i(\alc_i)$}{
    Choose one such path with minimum $\ell$\;
    Transfer $\itm_{\ell}$ from $i_{\ell}$ to $i_{\ell - 1}$\;
    Repeat from Phase 2\;
}
\lIf{$\alc$ is EQ1 \wrt market valuations}{\Return $\alc$}
}
\PhaseIII{Increase prices}{
$\beta \gets \min\{\mbb_i / \bb_i(\itm) \mid i \in \hierarchy_L, \itm \in \bigcup_{i \in \ags \setminus \hierarchy_L} \alc_i\}$\tcp*{\footnotesize{Note that $\mbb$ and $\bb$'s are based on subjective utilities $\avs$}}
\For{$\itm \in \bigcup_{i \in \hierarchy_L} \alc_i$}{
$p_\itm \gets \beta \cdot p_\itm$\;
}
Repeat from Phase 2 (i.e., go to line 4)\;
}
\end{algorithm}

\subsection{The Algorithm and Analysis}

\paragraph{The Algorithm.} We provide a constructive proof for \Cref{thm:eq1po} using \Cref{alg:eq1-po-two-valuations}. Similar to prior Fisher market based algorithms \cite{BKV18,freeman2019,murhekar2021fair}, the algorithm consists of three phases. The algorithm begins with an equilibrium allocation, which is Pareto optimal, and performs a local search by transferring goods locally (Phase 2) and changing the prices of the goods (Phase 3) in a way that the intermediate allocation $\alc$ and prices $\pp$ are always a Fisher market equilibrium. Eventually, the algorithm converges to an allocation that is EQ1 and fPO (as $\alc$ is an equilibrium allocation). 

In phase 1,  the algorithm begins with a utilitarian welfare maximizing allocation $\alc$ and sets the prices of goods equal to the utilities for which the allocated agents have. 

In phase 2, the algorithm takes the set of agents who have the lowest utility level $L = \{\ag \mid v_\ag(\alc_\ag) = \min_{h \in \ags} v_h(\alc_h)\}$, and searches through MBB paths for agents who ``violate'' the EQ1 condition. More formally, the algorithm searches for an agent $i_\ell$ via an MBB alternating path $i_0 \stackrel{\itm_1}\gets i_1 \stackrel{\itm_2}{\gets} \cdots \stackrel{\itm_\ell}{\gets} i_{\ell}$ such that $v_{i_\ell}(\alc_{i_\ell} \setminus \{\itm_{\ell}\}) > v_i(\alc_i)$. If such a path exists, the algorithm takes one of the shortest ones and transfers the good at the end of the path. Since the receiver of the good $\itm_\ell$ considers it MBB, the new allocation is still an equilibrium. 

If no such EQ1-violating MBB paths exist, in Phase 3, the algorithm finds the hierarchy $\hierarchy_L$ of the agents with the lowest level of utility $L$, who are reachable from $L$ via MBB alternating paths; and, it increases the prices of all goods allocated to $\hierarchy_L$ minimally such that a new MBB edge appears between $\hierarchy_L$ and $\ags \setminus \hierarchy_L$ the agents outside the hierarchy. Then, the algorithm continues with Phase 2 by searching for EQ1 violating MBB paths. This process will terminate in an allocation that is EQ1 \wrt market values and fPO \wrt subjective utilities allocation.

\paragraph{Sketch of the Analysis.} The analysis of the algorithm closely follows  the analysis of \citet{murhekar2021fair}, who show the existence of EQ1 and fPO allocations in the standard setting with heterogeneous subjective utilities. By invoking \Cref{alg:eq1-po-two-valuations} with $\mv \gets \avs$, we get the algorithm of \citet{murhekar2021fair} for the standard setting. 

Our key observation is that in \Cref{alg:eq1-po-two-valuations}, one can set the prices and define MBB goods for agents according to one valuation and ensure Pareto optimaly \wrt to this valuation while achieving EQ1 \wrt another valuation. We remark that the subjective utilities must be strictly positive, while the market valuations may have zero or positive values.  The requirement of nonzero subjective utilities ensures that in Phase 3, the price increase factor $\beta$ is bounded and no good has an unbounded bang-per-buck ratio. Recall that in the standard setting, which is a special case of our algorithm where $\avs = \mvs$, EQ1+PO is infeasible if there are zero utilities \cite{freeman2019}. 

We now provide a sketch of the proof of \Cref{thm:eq1po} by giving pointers to the analysis of \citet{freeman2019,murhekar2021fair}.

\eqonepo*{}

\begin{proof}
The allocation $\alc$ and price vector $\pp$ are always in an equilibrium, as shown in Lemma 2 of \citet{murhekar2021fair}. 

The only way \Cref{alg:eq1-po-two-valuations} terminates is by finding an EQ1 allocation. Thus, proving termination is the main challenge. A helpful observation is that the minimum market value never decreases, i.e., $\min_{\ag} v_\ag(\alc)$ is always non-decreasing, the proof of which follows similar to Lemma 6 of \citet{freeman2019}. This observation is used in deriving the following key arguments for the termination of the algorithm.
\begin{itemize}
    \item After $\poly(n, m)$ steps of Phase 2 and Phase 3, the set $L$ of agents with the lowest market value must change. This follows from arguments similar to Lemma 4 of \citet{murhekar2021fair} (for a more detailed proof see Lemma 13 of \citet{BKV18}).
    \item Further, if an agent $\ag$ ceases to be in the set $L$ at time-step $t_0$ with allocation $\alc^{t_0}$ and again becomes a part of $L$ at a later time-step $t_1$ with allocation $\alc^{t_1}$, it must hold that $v_\ag(\alc^{t_1}_\ag) > v_\ag(A^{t_0}_\ag)$. This closely follows Lemma 8 of \citet{murhekar2021fair}.
\end{itemize}

From the two arguments above combined, we have that the set $L$ changes after at most every $\poly(n, m)$ time-steps and it changes for at most $\poly(n, V)$ times. Phase 1 can simply be computed in $\poly(n, m)$ time. Therefore, the algorithm terminates in $O(\poly(n, m, V))$. 
\end{proof}

\section{Proof of \texorpdfstring{\Cref{thm:cake-PO-EQ-infinitesimal}}{Theorem~\ref{thm:cake-PO-EQ-infinitesimal}}}\label{app:cake-po}
\begin{proof}[Proof of \Cref{thm:cake-PO-EQ-infinitesimal}]
When the subjective density functions are piecewise constant, we can partition the cake into a set of $m$ intervals $\calI = \{I_1, \ldots, I_{m}\}$ such that for all agents $\ag \in \ags$, $u_\ag$ is a constant function over $I_j$ for all $\ag \in \ags$ and $j \in [m]$. In other words, we treat each such interval as a ``divisible good''.  We refer to the measure (length) of a piece of the cake $z$ by $\mu(z)$. 

\paragraph{Fisher Market Terminology.} We follow the Fisher market terminology defined in  \Cref{sec:fisher-markets}. For a fixed price vector $\pp$ over the intervals, the bang-per-buck ratio of agent $i$ for interval $I_j$ is defined as $\bb_i(I_j) = \frac{u_i(I_j) / \mu(I_j)}{p_{I_j}}$, where the numerator is the $i$'s density (total value for $I_j$ normalized by its length). Further, we call an interval $I_j$ a \emph{maximum bang-per-buck (MBB)} interval for agent $i$ if $\bb_i(I_j) = \max_{j \in [m]} \bb_i(I_j)$. For a price vector $\pp$, the MBB graph of $\pp$ is the bipartite graph with agents on one side and intervals on the other side with an edge $(i, I_j)$ existing iff $I_j$ is an MBB interval for $i$.

\paragraph{Attaining max egalitarian welfare \wrt market values subject to PO \wrt subjective utilities.} 
There are at most $2^{nm}$ MBB graphs (counting by fixing the subset of existing edges). Fix an MBB graph $G$ induced by a price vector $\pp$, where all intervals have at least one MBB edge (i.e., all intervals can be allocated). For each MBB graph $G$, we argue that we can find an allocation that maximizes the minimum value with respect to the heterogeneous market values among all allocations that respect the MBB graph. Specifically, this allocation must assign each interval $I_j$ only to agents for whom $I_j$ is MBB.

This allocation can be achieved by a straightforward application of the Dubins–Spanier theorem \cite{dubins1961cut}. According to this theorem, for each interval $I_j$, the space of valuation matrices (with respect to market values) restricted to the agents for whom $I_j$ is MBB is compact and convex. Summing the utility spaces over all intervals (padding the matrices with zeros for agents where the interval is not MBB) preserves compactness and convexity. Consequently, the space of valuation matrices over allocations that respect the MBB edges of $G$ remains compact and convex. Therefore, an allocation that maximizes the minimum utility for all agents is attainable.

Finally, by taking the maximum (egalitarian welfare) over all finitely many MBB graphs, we can achieve an allocation that maximizes the egalitarian welfare with respect to market values among all PO allocations \wrt subjective utilities.

\newcommand{\yy}{\ensuremath{\mathbf{y}}}
\paragraph{Existence of EQ \wrt market values and PO \wrt subjective utilities.} Similar to arguments for the max egalitarian welfare allocation, we can find an allocation $\xx$ satisfying the following.
\begin{enumerate}
    \item $\xx$ maximizes the egalitarian welfare among all PO allocations \wrt subjective utilities, i.e., 
    \[
    \xx \in \arg\max \{\min_{\ag \in \ags} \mv_\ag(\yy) \mid \yy ~\text{is PO \wrt}~ \avs\}.
    \]
    \item Subject to 1, $\xx$ minimizes the maximum utility. That is 
    \[
    \xx \in \arg\min \{\max_{\ag \in \ags} \mv_\ag(\yy) \mid \yy~\text{satisfies condition 1}\}.
    \]
    Let $v_{\min} = \min_{i} \mv_i(\xx)$ and $v_{\max} = \max_{i} \mv_i(\xx)$. 
    \item Subject to 1 and 2, $\xx$ minimizes the number of agents with the specified maximum utility of $v_{\max}$.
\end{enumerate}
We show that $\xx$ is EQ \wrt $\mvs$. Suppose,  by contradiction, that $\xx$ is not EQ. Let $\ags_{\max}$ be the set of agents with the maximum market value. Let $\calI_{\max} = \{I_j \mid \exists \ag \in \ags_{\max}, \mu(\xx_\ag \cap I_j) > 0\}$ be the set of intervals that a nonzero measure of which is allocated to $\ags_{\max}$.  We come to a contradiction by a case analysis.
\begin{itemize}
    \item \emph{Case 1.} Suppose there exists an agent $\ag' \in \ags \setminus \ags_{\max}$ that has an MBB edge to one of the intervals $I_j$ in $\calI_{\max}$ that is partially allocated to $\ag \in \ags_{\max}$. Then, we can transfer an infinitesimal fraction $\delta \in (0, \mu(\xx_{\ag'} \cap I_j))$ of $I_j$ from $\ag$ to $\ag'$ such that $\mv_{\ag'}(\xx_{\ag'} \cup \{\delta \cdot I_j\}) < v_{\max}$ and $\mv_{\ag}(\xx_{\ag} \setminus \{\delta \cdot I_j\}) > v_{\min}$ --- which is possible by the boundedness of the density function. This contradicts the condition 3 or 2 above.
    \item \emph{Case 2.} Otherwise, we have that all of $\calI_{\max}$ are only MBB to agents in $\ags_{\max}$. Since there are no MBB edges from $\ags \setminus \ags_{\max}$, we can increase the prices of all other intervals $\calI \setminus \calI_{\max}$, which are all allocated to agents in $\ags \setminus \ags_{\max}$ until an MBB edge from some agent in $\ags \setminus \ags_{\max}$ appears to an interval in $\calI_{\max}$. More specifically, we increase the prices by a factor of $\min\{\mbb_i / \bb_i(I_j) \mid \ag \in \ags \setminus \ags_{\max}, I_j \in \calI_{\max}\}$. After such a price increase, the allocation is still an equilibrium since (i) the bang-per-buck ratio for other intervals only decreases for agents in $\ags_{\max}$, and (ii) increasing the prices of all intervals in $\calI \setminus \calI_{\max}$ keeps the MBB edges from $\ags \setminus \ags_{\max}$ towards $\calI \setminus \calI_{\max}$ and the minimal increase in the prices ensures that the prior MBB edges remain MBB. Now, the condition of Case 1 is met, and we again come to a contradiction.
\end{itemize}

The requirement that subjective utilities be  strictly positive is necessary for the price rise argument in Case 2.

The proof can be slightly adjusted to work with market valuations that have zero densities. In Case 1, if the owner agent $\ag$ has a positive density over interval $I_j$, regardless of $\ag'$, the transfer is still valid and the fact that the market value of $\ag$ decreases is enough to get a contradiction. However, if $\mv_{\ag}$ has density $0$ over $I_j$, $\mv_{\ag}$ might remain the same. To address this, we add one more tie-breaking condition to the three conditions of $\xx$. The fourth condition is that $\xx$, subject to 1-3, minimizes the total measure of the cake allocated to $\ags_{\max}$ for which its owner  has zero market value. That is
\[
\xx \in \arg\min\left\{
\sum_{\ag\in\ags} \int_{x \in \yy_\ag} \mathbb{I}[\mv_{\ag}(x) = 0] \, d_x \; \mid \;  \yy~\text{satisfies conditions 1-3}\right
\}.
\]
This way, in the transfer of Case 1, if the market value of $\ag'$ is $0$, then the nonzero $\delta$ transfer that keeps $\mv_{\ag'}(\xx_{\ag'} \cup \{\delta \cdot I_j\}) < v_{\max}$ will strictly reduce the measure of the cake allocated to $\ags_{\max}$ with zero market value. This contradicts the fourth condition. Hence, $\xx$ is EQ \wrt heterogeneous market values and PO \wrt subjective density functions.
\end{proof}
\end{document}